\documentclass[12pt]{article}                
\usepackage{graphicx, xcolor}
\usepackage{psfrag}
\usepackage{amsmath, amssymb, amsthm}
\usepackage{mathptmx}
\usepackage{esvect, bm, bbm}
\usepackage{hyperref}
\usepackage{cleveref}
\usepackage{enumitem}
\usepackage{nicefrac}
\usepackage{subcaption}
\usepackage{tikz}
\usepackage{comment}

\hypersetup{colorlinks=true,
	linkcolor={blue!80!black},
	citecolor={green!50!black},
	urlcolor={red!50!black}
}

\renewcommand{\le}{\leqslant}

\renewcommand{\ge}{\geqslant}

\newcommand{\Ll}{\left}
\newcommand{\Rr}{\right}

\newcommand{\mcl}{\mathcal}
\newcommand{\td}{\tilde}
\renewcommand{\bar}{\overline}
\renewcommand{\d}{\mathrm{d}}

\newcommand{\N}{\mathbb{N}}
\newcommand{\Z}{\mathbb{Z}}

\newcommand{\R}{\mathbb{R}}
\newcommand{\Rd}{{\mathbb{R}^d}}

\newcommand{\E}{\mathbb{E}}

\newcommand{\al}{\alpha}
\newcommand{\ep}{\varepsilon}

\newcommand{\dr}{\partial}

\DeclareMathOperator{\Ber}{\operatorname{Ber}}

\newtheorem{theorem}{Theorem}[section]
\newtheorem{lemma}[theorem]{Lemma}
\newtheorem{proposition}[theorem]{Proposition}

\theoremstyle{remark}

\numberwithin{equation}{section}

\newcounter{Hequation}

\makeatletter
\g@addto@macro\equation{\stepcounter{Hequation}}
\makeatother

\begin{document}

\title{\vspace{-1cm}Breakdown of a concavity property of mutual information for non-Gaussian channels
}
\author{Anastasia Kireeva\thanks{\textsc{\tiny Department of Mathematics, ETH Zurich,  anastasia.kireeva@math.ethz.ch }} \and  Jean-Christophe Mourrat\thanks{\textsc{\tiny Department of Mathematics, ENS Lyon and CNRS, jean-christophe.mourrat@ens-lyon.fr}}}

\date{}
\maketitle
\vspace*{-0.7cm}
\begin{abstract}
Let $S$ and $\td S$ be two independent and identically distributed random variables, which we interpret as the signal, and let $P_1$ and $P_2$ be two communication channels. We can choose between two measurement scenarios: either we observe $S$ through $P_1$ and $P_2$, and also $\td S$ through $P_1$ and $P_2$; or we observe $S$ twice through $P_1$, and $\td{S}$ twice through $P_2$. In which of these two scenarios do we obtain the most information on the signal $(S, \td S)$? While the first scenario always yields more information when $P_1$ and~$P_2$ are additive Gaussian channels, we give examples showing that this property does not extend to arbitrary channels. 
As a consequence of this result, we show that the continuous-time mutual information arising in the setting of community detection on sparse stochastic block models is not concave, even in the limit of large system size. This stands in contrast to the case of models with diverging average degree, and brings additional challenges to the analysis of the asymptotic behavior of this quantity. 
\end{abstract}

%
%
%
%

\section{Introduction}
Let $P_S$ be a probability measure with finite support $\mcl S$, and let $S$ be a random variable sampled according to $P_S$, which we think of as a signal. A \emph{communication channel} over~$\mcl S$, or more simply a \emph{channel}, is a family of probability measures $(P(\cdot \mid s))_{s \in \mcl S}$ over $\R^d$ for some integer $d \ge 1$, which we view as a conditional probability distribution over $\Rd$ given~$S$. Let $f : \mcl S \to \Rd$, and let $W$ be a standard $d$-dimensional Gaussian random vector independent of $S$. The conditional law, given $S$, of the random variable
\begin{equation}
\label{e.def.Gaussian.channel}
X := f(S) + W
\end{equation}
defines a channel. We call any channel that can be constructed in this way  a \emph{Gaussian channel}. The information-theoretic quantities studied in this paper are invariant under bijective bimeasurable transformations of the channel output; in particular, there is no loss of generality in assuming the covariance matrix of the noise term in \eqref{e.def.Gaussian.channel} to be the identity. For random variables $X$ and $Y$ defined on the same probability space, we denote by $I(X;Y)$ their mutual information, that is, 
\begin{equation*}  
I(X;Y) := \E \Ll[ \log \Ll(\frac{\d P_{(X,Y)}}{\d P_X \otimes \d P_Y}(X,Y) \Rr) \Rr] ,
\end{equation*}
where $P_{(X,Y)}$, $P_X$ and $P_Y$ are the laws of $(X,Y)$, $X$ and $Y$ respectively.

Let $P_1$ and $P_2$ be two channels over $\mcl S$. Conditionally on $S$, we sample $X_1$, $X'_1$, $X_2$ and $X'_2$ independently, with $X_1, X'_1$ sampled according to $P_1(\cdot \mid S)$, and $X_2$, $X'_2$ sampled according to $P_2(\cdot \mid S)$. We consider the following question.
\begin{equation}
\tag{Q1}
\label{e.question}
\text{Do we have } \ I(S;(X_1,X'_1)) + I(S;(X_2,X'_2)) \le 2 I(S;(X_1,X_2)) \ ?
\end{equation}

\begin{figure}[tbp]
\centering
\begin{subfigure}{0.45\textwidth}
    \centering
    \tikzset{every node/.style={font=\large},
    every path/.style={line width=0.25mm},
    }
    \begin{tikzpicture}
    \def\height{2.5}
    \node (S) at (0,\height) {$S$};
    \node (X) at (-1,0) {$X_1$};
    \node (Y) at (1,0) {$X_2$};
    \draw [->] (S) -- (X) node[midway,above left,scale=0.7] {$P_1$};
    \draw [->] (S) -- (Y) node[midway,above right,scale=0.7] {$P_2$};
    
    \node (S_tilde) at (3.5,\height) {$\td S$};
    \node (X_tilde) at (2.5,0) {$\tilde X_1$};
    \node (Y_tilde) at (4.5,0) {$\tilde X_2$};
    \draw [->] (S_tilde) -- (X_tilde) node[midway,above left,scale=0.7] {$P_1$};
    \draw [->] (S_tilde) -- (Y_tilde) node[midway,above right,scale=0.7] {$P_2$};
    \end{tikzpicture}
\caption{Scenario 1: We observe the signal and its independent copy twice through both channels $P_1$ and $P_2$.}
\end{subfigure}
\hspace{5mm}
\begin{subfigure}{0.45\textwidth}
    \centering
    \tikzset{every node/.style={font=\large},
    every path/.style={line width=0.25mm},
    }
    \begin{tikzpicture}
    \def\height{2.5}
    \node (S) at (0,\height) {$S$};
    \node (X) at (-1,0) {$X_1$};
    \node (Y) at (1,0) {$X_1^\prime$};
    \draw [->] (S) -- (X) node[midway,above left,scale=0.7] {$P_1$};
    \draw [->] (S) -- (Y) node[midway,above right,scale=0.7] {$P_1$};
    
    \node(S_tilde) at (3.5,\height) {$\td S$};
    \node (X_tilde) at (2.5,0) {$\td X_2$};
    \node (Y_tilde) at (4.5,0) {$\td X_2^\prime$};
    \draw [->] (S_tilde) -- (X_tilde) node[midway,above left,scale=0.7] {$P_2$};
    \draw [->] (S_tilde) -- (Y_tilde) node[midway,above right,scale=0.7] {$P_2$};
    \end{tikzpicture}
\caption{Scenario 2: We observe the signal twice through channel $P_1$ and its independent copy twice through channel $P_2$.}
\end{subfigure}
\caption{Does the scenario on the left side give us more information about $(S,\td S)$ than the scenario on the right side?}
\label{fig:sampling_procedure}
\end{figure}
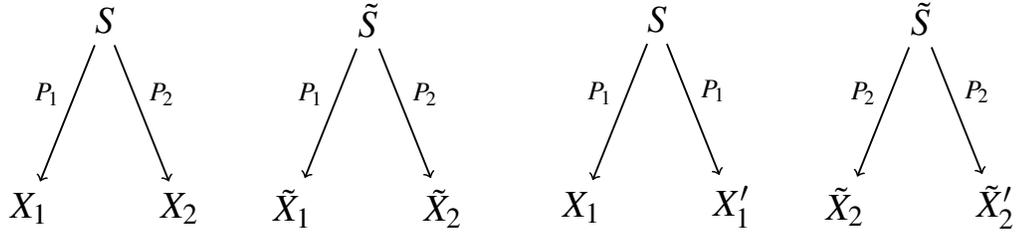

A possibly more intuitive way to ask this question, following the phrasing in the abstract, is displayed in Figure~\ref{fig:sampling_procedure}, where we denote by $(\td S, \td X_1, \td X'_1, \td X_2, \td X'_2)$ an independent copy of $(S, X_1, X'_1, X_2, X'_2)$.
As will be seen below, the answer to this question is positive whenever $P_1$ and $P_2$ are Gaussian channels. However, we will show that the answer to this question is actually negative if $P_1$ and $P_2$ can be arbitrary channels. In fact, our counterexamples are even such that 
\begin{equation*}  
\min \Ll( I(S;(X_1,X'_1)) , I(S;(X_2,X'_2)) \Rr) > I(S;(X_1,X_2)).
\end{equation*}
While we find this question interesting on its own, we are also motivated by its implications in the context of community detection problems. We consider the setting of the stochastic block model \cite{Finberg, Holland, Yuchung, Harrison}, sometimes also called the planted partition model \cite{Boppana, Bui, Dyer} or the inhomogeneous random graph model \cite{Bollobas}. In the case of two communities, this model is defined as follows. First, we independently attribute each individual to one of the two possible communities. Next, independently for each pair of individuals, we draw an edge between these two individuals with probability $d_{\mathrm{in}}/N$ if the two individuals belong to the same community, or with probability $d_{\mathrm{out}}/N$ if the two individuals belong to different communities, where $N$ is the total number of individuals. We are then shown the resulting graph, but not the underlying community structure, which we aim to reconstruct. The choice of scaling for the link probabilities ensures that the average degree of a node remains bounded as $N$ tends to infinity. 

This problem has received considerable attention. An early contribution is the very inspiring work of \cite{DKMZ}, which relies on deep non-rigorous statistical physics arguments. 
In the case when the individuals are equally likely to belong to one or the other community, it was shown in \cite{Massoulie, Mossel_reconstruction, Mossel_detection} that one can recover meaningful information on the underlying community structure if and only if $(d_{\mathrm{in}} - d_{\mathrm{out}})^2 > 2 (d_{\mathrm{in}} + d_{\mathrm{out}})$; and in this case, there exists an efficient algorithm for doing so. 

A more refined question consists in studying the asymptotic behavior of the mutual information between the observed graph and the community structure, in the limit of large~$N$. When $d_{\mathrm{in}} < d_{\mathrm{out}}$, this problem was resolved in \cite{Abbe_disassortative, CKPZ}. The case when ${d_{\mathrm{in}} > d_{\mathrm{out}}}$ is more challenging and was only resolved very recently in~\cite{yu2022ising}; we also refer to \cite{Abbe_assortative, kanade2016global, mossel2016belief, mossel2016local} for earlier work on this. The core of the argument of \cite{yu2022ising} is to show that there is a unique fixed point to a certain belief-propagation (BP) distributional recursion. 

For situations with four or more communities, the problem becomes more complicated, and there exist choices of parameters for which this BP fixed-point equation admits more than one solution \cite{gu2023uniqueness}. In these cases, a strategy in the spirit of that deployed in \cite{yu2022ising} therefore cannot be adapted in a straightforward way, and further work is necessary.

An alternative approach to the problem of identifying the asymptotic behavior of the mutual information between the observed graph and the community structure has been initiated in \cite{sparse_PDE, sparse_prob}. The gist of the approach is to identify the limit mutual information as the solution to a certain partial differential equation (PDE). This technique allowed for the asymptotic analysis of the mutual information of a very large class of models involving Gaussian channels \cite{JC_HB_FinRank}; see also \cite{chen2022hamilton, chen2021limiting, chen2022hamilton2,  JC_matrix, JC_HJ}. Using other approaches, a number of special cases had been solved earlier in \cite{barbier2016, barbier2019adaptive, barbier2017layered, kadmon2018statistical, Lelarge, lesieur2017statistical,  luneau2021mutual, luneau2020high, mayya2019mutual, miolane2017fundamental,  reeves2020information, reeves2019geometry}. As shown in \cite{barbier2019dense, Deshpande, Lelarge}, a Gaussian equivalence property ensures that these results also allow us to identify the asymptotic behavior of the mutual information of the community detection problem in regimes in which the average degree of a node diverges with the system size. 

In the approach taken up in \cite{JC_HB_FinRank, sparse_prob},  one can leverage a certain regularity property of the mutual information to obtain a lower bound on the limit mutual information in terms of the solution to the PDE. This is similar to the results obtained in \cite{JC_NC, JC_upper} in the context of spin glasses. In order to show the matching upper bound, a central ingredient of the approach taken up in \cite{JC_HB_FinRank} is the observation that the mutual information is a \emph{concave} function of the signal-to-noise ratios of the various observations considered for the resolution of the problem. For the community detection problem, if the mutual information studied in \cite{sparse_prob} happened to be concave in its parameters, we would be optimistic that the approach of \cite{JC_HB_FinRank} would be adaptable to this setting, and thus would allow us to obtain the matching upper bound. However, we show here that the mutual information is in fact \emph{not} a concave function of its parameters. We find this surprising given that this concavity property does hold for the problems with Gaussian channels considered in \cite{JC_HB_FinRank} and elsewhere. We derive this breakdown of concavity as a consequence of the fact that the answer to Question~\ref{e.question} is negative in general. Precisely, we will show that, although the Hessian of the mutual information only contains nonpositive entries, we can witness a breakdown of concavity that scales as $(d_{\mathrm{in}} - d_{\mathrm{out}})^6$ in the regime of small $|d_{\mathrm{in}} - d_{\mathrm{out}}|$. We are also surprised by the relatively high exponent $6$ appearing here, suggesting a rather subtle deviation from concavity in the regime of small $|d_{\mathrm{in}} - d_{\mathrm{out}}|$. 

Had the mutual information been concave in its parameters, we would presumably have been able to represent the solution to the relevant PDE as a saddle-point variational problem, using a version of the Hopf formula  (see \cite{JC_HB_FinRank}, and also \cite{sparse_PDE} for a proof of a related variational formula under different assumptions). Given that this concavity property is in fact invalid, we tend to think that there will not be any reasonable way to represent the limit mutual information of community detection as a variational problem, unlike the situation with Gaussian channels. In the context of spin glasses, this point is discussed more precisely in \cite[Section~6]{JC_NC}. 

The rest of the paper is organized as follows. In Section~\ref{s.results}, we show that the answer to Question~\ref{e.question} is positive for Gaussian channels, and construct counterexamples in general. We pay special attention to the case of Bernoulli channels with very low signal-to-noise ratios, as these examples will be fundamental to subsequent considerations concerning the community detection problem. In Section~\ref{s.multidimensional}, we focus on Gaussian channels and explore variants of the inequality appearing in Question~\ref{e.question} that involve more than two channels. In Section~\ref{s.community}, we turn to the setting of community detection, for the stochastic block model with two communities. We use the results of Section~\ref{s.results} to show that the mutual information is not a concave function of its parameters, even after we pass to the limit of large system size. 

%
%
%
%
\section{Answers to Question~\ref{e.question}}
\label{s.results}

We start by providing a positive answer to Question~\ref{e.question} in the case of Gaussian channels.
\begin{proposition}[Mixing Gaussian channels yields more information]
\label{p.gaussians}
If $P_1$ and~$P_2$ are Gaussian channels, then the answer to Question \ref{e.question} is positive. 
\end{proposition}
\begin{proof}
The proof of Proposition~\ref{p.gaussians} is based on remarkable identities involving derivatives of the mutual information with respect to the signal-to-noise ratio. In particular, the first-order derivative of the mutual information is half of the minimal mean-square error, as was explained in \cite{guo2005mutual} and extended to the matrix case in~\cite{lamarca2009linear, payaro2009hessian, reeves2018mutual}. Here we will rely on the calculation of second-order derivatives of the mutual information, which already appeared in \cite{guo2011estimation, lamarca2009linear, payaro2009hessian}.

By definition of Gaussian channels, for each $i \in \{1,2\}$, there exists a mapping $f_i : \mcl S \to \R^{d_i}$ such that the channel $P_i$ can be represented as 
\begin{equation*}  
S \mapsto f_i(S) + W_i,
\end{equation*}
where $W_1$, $W_2$ are independent standard Gaussians, independent of $S$, of dimension~$d_1$ and $d_2$ respectively. For every $i \in \{1,2\}$ and $t_i \ge 0$, we define
\begin{equation*}  
X_i(t_i) := \sqrt{t_i} f_i(S) + W_i,
\end{equation*}
as well as
\begin{equation*}  
\mcl I (t_1,t_2) := I(S; (X_1(t_1), X_2(t_2))).
\end{equation*}
Since the mapping $s \mapsto (s,f_1(s),f_2(s))$ is injective, we have
\begin{equation*}  
\mcl I(t_1,t_2) = I((S,f_1(S),f_2(S)); (X_1(t_1), X_2(t_2))).
\end{equation*}
We can therefore replace the signal $S$ by $(S,f_1(S), f_2(S))$ if desired, and apply \cite[Theorem~3]{lamarca2009linear} or \cite[Theorem~5]{payaro2009hessian} with $H$ chosen to be the identity matrix and $P$ chosen to be a diagonal matrix with $d_1$ entries at $\sqrt{t_1}$ and $d_2$ entries at $\sqrt{t_2}$. The conclusion of these theorems is that the function $\mcl I$ is jointly concave in $(t_1,t_2)$. In particular,
\begin{equation*}  
\mcl I(2,0) + \mcl I(0,2) \le 2 \mcl I(1,1).
\end{equation*}
Recalling that 
\begin{equation*}  
\mcl I(1,1) = I(S;(X_1(1),X_2(1))) = I(S;(X_1,X_2)),
\end{equation*}
Proposition~\ref{p.gaussians} will be proved once we verify that 
\begin{equation}
\label{e.to.verify}
\mcl I(2,0) = I(S;(X_1,X'_1)) \quad \text{ and } \quad \mcl I(0,2) = I(S;(X_2,X'_2)).
\end{equation}
We fix $i \in \{1, 2\}$, let $W_i'$ be a $d_i$-dimensional standard Gaussian independent of $(S,W_i)$, and use it to represent $X_i'$ as
\begin{equation*}  
X'_i = f_i(S) + W_i'.
\end{equation*}
We define
\begin{equation*}  
Z_i := \frac {X_i + X'_i}{\sqrt{2}} = \sqrt{2t} f_i(S) + \frac{W_i + W'_i}{\sqrt{2}}, 
\end{equation*}
and
\begin{equation*}  
D_i := X_i - X'_i = W_i - W_i'.
\end{equation*}
Using that the the map $(x,y) \mapsto ((x+y)/\sqrt{2}, x-y)$ is bijective and the chain rule, we can write
\begin{equation*}  
I(S;(X_i,X'_i)) = I(S;(Z_i,D_i)) = I(S;D_i) + I(S;Z_i \mid D_i).
\end{equation*}
The random variables $S$ and $D_i$ being independent, the first term on the right side of this identity vanishes. We also observe that the pair $(W_i,W'_i)$ is independent of~$S$, and moreover, the Gaussian random variables $W_i+W_i'$ and $W_i-W_i'$ are independent. This implies that the random variables $(S,W_i+W_i',W_i - W_i')$ are independent, and thus that $D_i$ is independent of the pair $(S,Z_i)$. The previous display therefore simplifies into
\begin{equation*}  
I(S;(X_i,X'_i)) = I(S;Z_i).
\end{equation*}
Since the pairs $(S,X_i(2))$ and $(S,Z_i)$ have the same law, this is \eqref{e.to.verify}. 
\end{proof}
We now turn to showing that Proposition~\ref{p.gaussians} does not generalize to non-Gaussian channels. Before doing so, we record a simple observation allowing to simplify the question somewhat.
\begin{lemma}
\label{l.i.identity}
Let $S$ be a random variable with finite support $\mcl S$, let $P_1, P_2$ be two communication channels over $\mcl S$, and conditionally on $S$, let $(X_1,X_1', X_2, X_2')$ be independent random variables, with $X_1, X_1'$ sampled according to $P_1(\cdot \mid S)$ and $X_2, X'_2$ sampled according to $P_2(\cdot \mid S)$. For every $i,j \in \{1,2\}$, we have
\begin{equation}  
\label{e.i.identity}
I(S;(X_i;X'_j)) = I(S;X_i) + I(S;X_j) - I(X_i;X'_j).
\end{equation}
\end{lemma}
\begin{proof}
By the chain rule,
\begin{align*}  
I(S;(X_i,X'_j)) & = I(S;X_i) + I(S;X'_j \mid X_i) \\
& = I(S;X_i) + I(X'_j ; (S, X_i)) - I(X_i;X'_j) \\
& = I(S;X_i) + I(S;X'_j) + I(X_i; X'_j \mid S) - I(X_i;X'_j).
\end{align*}
Conditionally on $S$, the random variables $X_i$ and $X'_j$ are independent. It thus follows that ${I(X_i, X'_j \mid S)} = 0$, and we obtain \eqref{e.i.identity}. 
\end{proof}
A direct consequence of Lemma~\ref{l.i.identity} is that 
\begin{equation*}  
2 I(S;(X_1,X_2)) - I(S;(X_1,X'_1)) - I(S;(X_2,X'_2))  = I(X_1, X'_1) + I(X_2, X'_2) - 2 I(X_1, X_2),
\end{equation*}
And in particular, Question~\ref{e.question} can be rephrased as:
\begin{equation}
\tag{Q2}
\label{e.question.2}
\text{Do we have } \     2 I(X_1, X_2) \le I(X_1, X'_1) + I(X_2, X'_2) \ ?
\end{equation}
For every $p \in [0,1]$, we write $\Ber(p) := p\delta_1 + (1-p) \delta_0$ for the law of a Bernoulli random variable of parameter $p$. 
For our counterexamples, we assume that $S$ is a $\Ber(1/2)$ random variable, and we consider channels of the following form, for different choices of $p_0, p_1, q_0, q_1 \in [0,1]$:
\begin{equation*}  
P_1(\cdot \mid s) = \Ber(p_s)\quad 
\text{ and }
\quad
P_2(\cdot \mid s) = \Ber(q_s)\quad \quad (s \in \{0,1\}).
\end{equation*}
Already for the choice of $p_0 = 1/2$, $p_1 = 0$, $q_0 = 0$, $q_1 = 1/2$, we find that 
\begin{equation*}  
I(X_1;X_2) = \frac 5 2\log(2) - \frac 3 2 \log(3) \simeq 0.0849,
\end{equation*}
while
\begin{equation*}  
I(X_1;X'_1) = I(X_2;X'_2) = \log(2) + \frac 5 8 \log(5) - \frac 3 2 \log(3) \simeq 0.0511.
\end{equation*}

\begin{figure}[tbp]
    \centering
    \includegraphics[scale=0.38]{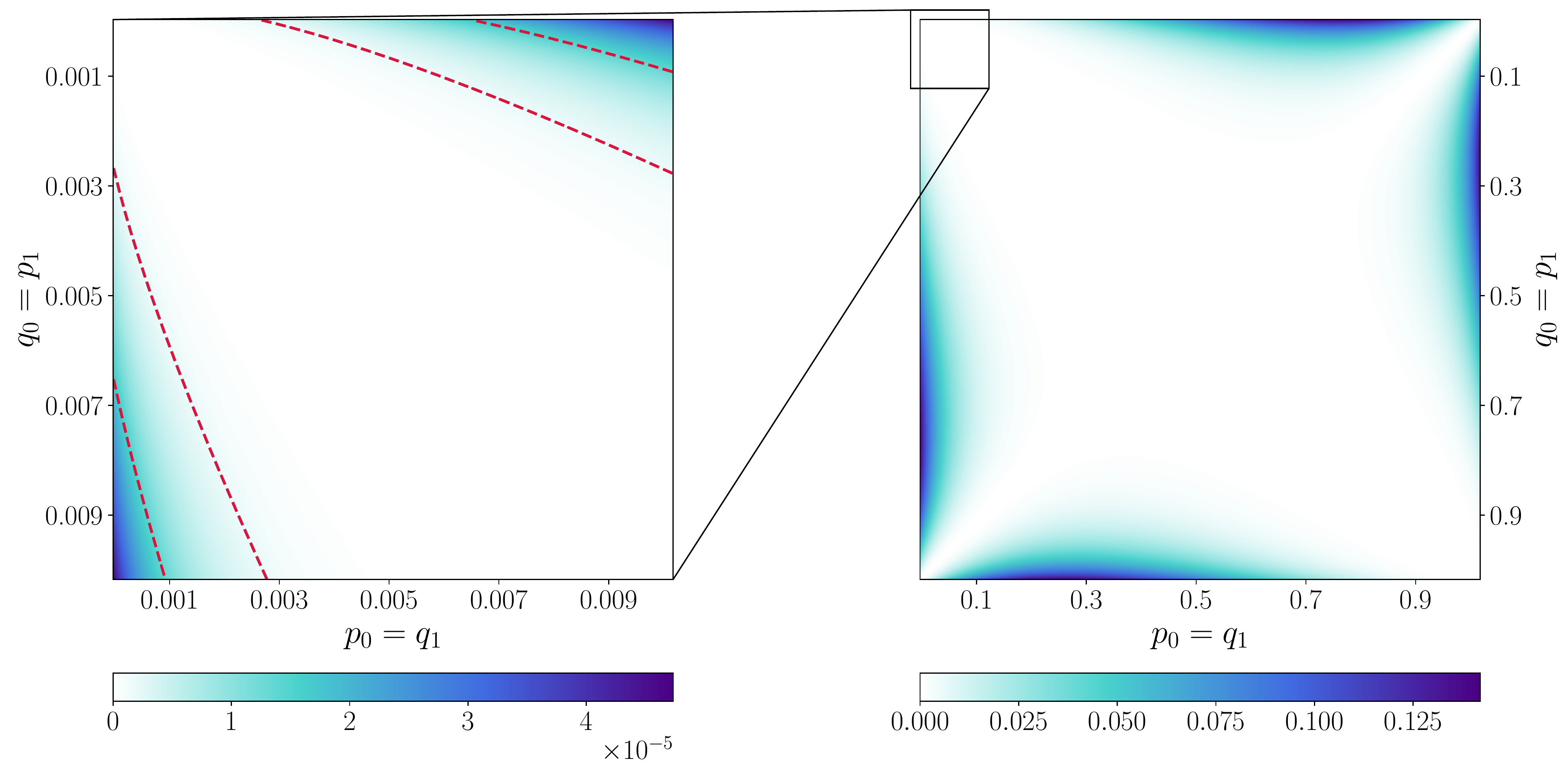}
    \caption{Value of the $2 I (X_1, X_2) - I(X_1, X_1^\prime) - I(X_2, X_2^\prime)$ for the setting of Bernoulli channels (see description in the text). The value is represented by color, the larger values correspond to darker color. Left: the regime of small $p_0, p_1$. Red dashed lines are countour lines of $\nicefrac{(p_0-p_1)^6}{(p_0 + p_1)^4}$. Right: general $p_0, p_1 \in [0,1]$.  }
    \label{fig:heatmap_bernoulli_symmetric}
\end{figure}
In particular, this leads to a counterexample to the inequalities appearing in Questions~\ref{e.question} and \ref{e.question.2}. In fact, for Bernoulli channels with $p_0 = q_1$ and $p_1=q_0$, we observe numerically that there are large regions of values of $p_0, p_1$ where the inequality in Question~\ref{e.question.2} does not hold---in fact, probably all values with $p_0 \neq p_1$, see \Cref{fig:heatmap_bernoulli_symmetric}. Perhaps surprisingly in view of Proposition~\ref{p.gaussians} and its proof, the next proposition shows that the inequalities in Questions~\ref{e.question} and \ref{e.question.2} can be violated even in regimes of small signal-to-noise ratio. This class of examples will be particularly relevant in the context of community detection discussed later.
\begin{proposition}
\label{p.taylor.expansion}
Assume that $S$ is a $\Ber(1 / 2)$ random variable, that we sample $X_1$, $X'_1$ according to $P_1(\cdot \mid S)$ and $X_2$, $X'_2$ according to $P_2(\cdot \mid S)$, where $P_1$ and $P_2$ are such that, for some $p_0, p_1, q_0, q_1 \ge 0$ and $\varepsilon > 0$,
\begin{equation*}  
P_1(\cdot \mid s) = \Ber(\ep p_s)\quad 
\text{ and }
\quad
P_2(\cdot \mid s) = \Ber(\ep q_s)\quad \quad (s \in \{0,1\}).
\end{equation*} 
If $p_0 = q_1$, $p_1 = q_0$, then
\begin{equation*}
2 I(X_1, X_2) - I(X_1, X'_1) - I(X_2, X'_2) \ge \frac{\ep^2(p_0 - p_1)^6}{6 (p_0 + p_1)^4} + o(\ep^2) \qquad (\ep \to 0).
\end{equation*}
In particular, the inequalities in Questions~\ref{e.question} and \ref{e.question.2} are false whenever $p_0 \neq p_1$ and $\ep > 0$ is sufficiently small.
\end{proposition}

\begin{proof}
For unconstrained values of $p_0, p_1, q_0, q_1 \ge 0$, the definition of mutual information  yields that
\begin{equation*}
    \begin{split}
        I(X_1, X_2) &= \log \left ( \frac{\frac{1}{2}(1 - \varepsilon p_0) (1 - \varepsilon q_0) + \frac{1}{2}(1 - \varepsilon p_1) (1 - \varepsilon q_1)}{\frac{1}{4} (2 - \varepsilon p_0 - \varepsilon p_1) (2 - \varepsilon q_0 -\varepsilon q_1)} \right)  \\
        &\qquad \cdot \frac{1}{2} \left( (1 - \varepsilon p_0) (1 - \varepsilon q_0) + (1 - \varepsilon p_1) (1 - \varepsilon q_1)\right)  \\
        &\quad 
        + \log \left(\frac{\frac{1}{2}\varepsilon p_0 (1 - \varepsilon q_0) + \frac{1}{2} \varepsilon p_1 (1 - \varepsilon q_1)}{\frac{1}{4} \varepsilon (p_0 +  p_1) (2 - \varepsilon q_0 -\varepsilon q_1)}  \right) \cdot \frac{1}{2}\varepsilon ( p_0 (1 - \varepsilon q_0) + p_1 (1 - \varepsilon q_1)) \\
        &\quad + \left(\frac{\frac{1}{2}\varepsilon q_0 (1 - \varepsilon p_0) + \frac{1}{2} \varepsilon q_1 (1 - \varepsilon p_1)}{\frac{1}{4}\varepsilon (2 - \varepsilon p_0 -\varepsilon p_1) ( q_0 +  q_1) }  \right) \cdot \frac{1}{2}\varepsilon (q_0 (1 - \varepsilon p_0) + q_1 (1 - \varepsilon p_1)) \\
        &\quad + \log \left(\frac{\frac{1}{2} \varepsilon^2 p_0 q_0 + \frac{1}{2} \varepsilon^2 p_1  q_1}{\frac{1}{4} \varepsilon^2 ( p_0 + p_1) ( q_0 +  q_1)}  \right) \cdot \frac{1}{2} \varepsilon^2 (p_0 q_0 + p_1  q_1) \\
        &=: I_1 + I_2 + I_3 + I_4.
    \end{split}
\end{equation*}
The Taylor expansions of these terms are given by
\begin{equation*}
    \begin{split}
        I_1 &= \frac{1}{4} \varepsilon^2 (p_0 - p_1) (q_0 - q_1)+ o(\varepsilon^2),\\
        I_2 = I_3 &= - \frac{1}{4} \varepsilon^2 (p_0 - p_1) (q_0 - q_1) + o(\varepsilon^2), \\
        I_4 &= \frac{1}{2} \varepsilon^2 (p_0 q_0+p_1 q_1) \log \left(\frac{2 \left(p_0 q_0+p_1 q_1\right)}{\left(p_0+p_1\right) \left(q_0+q_1\right)}\right)+o(\varepsilon^2) .
    \end{split}
\end{equation*}
Hence we obtain 
\begin{equation*}
\begin{split}
    I(X_1, X_2) = &- \frac{1}{4} \varepsilon^2 (p_0 - p_1) (q_0 - q_1) \\
        & + \frac{1}{2} \varepsilon^2 \log \left ( 1 + \frac{(p_0 - p_1)(q_0 - q_1)}{(p_0 + p_1)(q_0 + q_1)}\right) (p_0 q_0 + p_1 q_1) + {o}(\varepsilon^2).
\end{split}
\end{equation*}
In the case when $p_0 = q_0$ and $p_1 = q_1$, we get
\begin{equation*}
    \begin{split}
        I(X_1, X'_1) &= I(X_2, X'_2) \\
        &= -\frac{1}{4} \varepsilon^2 (p_0 - p_1)^2 + \frac{1}{2} \log \left( 1 + \frac{(p_0-p_1)^2 }{(p_0 + p_1)^2} \right)  (p_0^2 + p_1^2) + o(\varepsilon^2),
    \end{split}
\end{equation*}
while in the considered case when $p_0 = q_1$ and $p_1 = q_0$, 
\begin{equation*}
    I(X_1, X_2) = \frac{1}{4} \varepsilon^2 (p_0 - p_1)^2 + p_0 p_1 \log \left( 1 - \frac{(p_0-p_1)^2 }{(p_0 + p_1)^2} \right)   + o(\varepsilon^2).
\end{equation*}
Rescaling the mutual information by $\varepsilon^2$, we obtain that
\begin{equation}\label{e.mutual_info_lbound}
\begin{split}
    &\frac{1}{\varepsilon^2}\left ( 2 I(X_1, X_2) - I(X_1, X'_1) - I(X_2, X'_2) \right) \\
    &\quad = (p_0 - p_1)^2 + 2 p_0 p_1 \log \left ( 1 - \frac{(p_0 - p_1)^2}{(p_0 + p_1)^2} \right) \\
    &\qquad -  (p_0^2 + p_1^2) \log \left( 1 + \frac{(p_0-p_1)^2 }{(p_0 + p_1)^2} \right) + o(1) 
\end{split}
\end{equation}
Denoting $t := \nicefrac{(p_0 - p_1)^2}{(p_0 + p_1)^2} \in [0,1]$, we can rewrite the above identity as
\begin{equation*}
    \begin{split}
        &\frac{1}{\varepsilon^2 (p_0+p_1)^2 }\left ( 2 I(X_1, X_2) - I(X_1, X'_1) - I(X_2, X'_2) \right) + o(1) \\ &\quad = t + \frac{1 - t}{2} \log  ( 1 - t) - \frac{1 + t}{2} \log(1 + t) =: g(t).
    \end{split}
\end{equation*}
We observe that $g(0)=0$, or equivalently, this difference is zero when $p_0 = p_1$. We compute the derivative of $g$ with respect to $t$ and obtain that
\begin{equation*}
    g'(t) = - \frac{1}{2} (\log(1 - t) + \log(1 + t)) = -\frac{1}{2}\log(1 - t^2) \ge \frac {t^2}{2},
\end{equation*}
so $g(t) \ge \frac{t^3}{6}$ for every $t \in [0,1]$. Substituting back $t$ yields that
\begin{equation*}
    2 I(X_1, X_2) - I(X_1, X'_1) - I(X_2, X'_2) \ge  \frac{\varepsilon^2 (p_0 - p_1)^6 }{6 (p_0 + p_1)^4 } + o(\epsilon^2),
\end{equation*}
as desired.
\end{proof}

%
%
%
%

\section{Positive semidefinite kernels in the Gaussian case}
\label{s.multidimensional}

For arbitrary random variables $(Z_i)_{1 \le i \le n}$, one may wonder whether $(I(Z_i;Z_j))_{1 \le i,j \le n}$ is a positive semidefinite matrix; a negative answer to this question was provided in \cite{jakobsen2014mutual}. In our setting, consider multiple channels $(P_1,\ldots, P_n)$ over $\mcl S$, and conditionally on $S$, denote by $(X_i,X'_i)_{1 \le i \le n}$ conditionally independent random variables with $X_i$, $X'_i$ distributed according to $P_i(\cdot \mid S)$. One could ask:
\begin{equation}
\tag{Q3}
\label{e.question.matrix}
\text{Is the matrix } (I(X_i;X'_j))_{1 \le i, j \le n}  \text{ positive semidefinite?}
\end{equation}
We find that this is a natural question on its own; as will be seen below, it also arises naturally in the study of the continuous-time mutual information discussed below in relation with the problem of community detection.
If the answer to Question~\ref{e.question.matrix} were positive, then it would mean that the mapping $(P_i, P_j) \mapsto I(X_i; X'_j)$ defines a positive semidefinite kernel over the space of channels. Notice that Question~\ref{e.question.2} can be rephrased as
\begin{equation*}  
\text{Do we have } 
\begin{pmatrix}  
1 \\ -1 
\end{pmatrix}
\cdot
\begin{pmatrix}  
I(X_1;X'_1) & I(X_1; X'_2) \\ I(X_2; X'_1) & I(X_2;X'_2)
\end{pmatrix}
\begin{pmatrix}  
1 \\ -1 
\end{pmatrix} \ge 0 \ ?
\end{equation*}
Since we identified examples for which the inequality in Question~\ref{e.question.2} is violated, it follows that the answer to Question~\ref{e.question.matrix} is also negative in general. We do not know whether the answer to Question~\ref{e.question.matrix} is positive for Gaussian channels. Roughly speaking, the next proposition states that the answer to Question~\ref{e.question.matrix} is positive for Gaussian channels in the low signal-to-noise regime.
\begin{proposition}[psd kernel for Gaussian channels]
\label{p.psd}
Let $n \ge 1$ be an integer. For every $i \in \{1,\ldots, n\}$, let $f_i : \mcl S \to \R^{d_i}$, and let $(W_i, W_i')_{1 \le i \le n}$ be independent standard Gaussian random vectors, independent of the signal $S$, with $W_i$ and $W_i'$ of dimension $d_i$.  For every $i \in \{1,\ldots, n\}$ and $t \ge 0$, we define
\begin{equation*}  
X_i(t) := \sqrt{t} f_i(S) + W_i \quad \text{ and } \quad X'_i(t) := \sqrt{t} f_i(S) + W_i'.
\end{equation*}
For every $i,j \in \{1,\ldots, n\}$, we have
\begin{equation}
\label{e.psd.lim}
\lim_{t \to 0} \, t^{-2}\,  I(X_i(t); X'_j(t)) = \big|\E \big[ (f_i(S) - \E[f_i(S)]) (f_j(S) - \E[f_j(S)])^*\big] \big|^2,
\end{equation}
where the superscript $^*$ denotes the transpose operator, and the norm $|\cdot|$ over matrices is the Frobenius norm. 
Moreover, the matrix 
\begin{equation}
\label{e.psd}
\Ll(\big|\E \big[ (f_i(S) - \E[f_i(S)]) (f_j(S) - \E[f_j(S)])^*\big] \big|^2\Rr)_{1 \le i,j \le n}
\end{equation}
 is positive semidefinite.
\end{proposition}
\begin{proof}
The proof is again based on the fundamental identities derived in \cite{guo2005mutual, lamarca2009linear, payaro2009hessian}. In order to ligthen the notation, we define, for every $i \in \{1,\ldots, n\}$ and $s \in \mcl S$,
\begin{equation*}  
\bar f_i(s) := f_i(s) - \E \Ll[ f_i(S) \Rr] .
\end{equation*}
Recalling that we assume the state space $\mcl S$ of the signal $S$ to be finite, one can check that the mapping $t \mapsto I(S;X_a(t))$ is infinitely differentiable. The I-MMSE relation from \cite{guo2005mutual} yields that 
\begin{equation}
\label{e.first.order}
\dr_t I(S;X_i(t))_{\mid t = 0} = \frac 1 2 \E \Ll[ \Ll| \bar f_i(S) \Rr| ^2 \Rr] ,
\end{equation}
while \cite[Theorem~5]{payaro2009hessian} or the proof of \cite[Theorem~3]{lamarca2009linear} imply that 
\begin{equation}
\label{e.second.order}
\dr_t^2 I(S;X_i(t))_{\mid t = 0}  = \frac 1 2 \Ll| \E \Ll[ \bar f_i(S) \bar f_i(S)^* \Rr] \Rr|^2.
\end{equation}
Since the choice of $f_i$ is arbitrary, the identities \eqref{e.first.order} and \eqref{e.second.order} also imply that
\begin{equation}  
\label{e.first.order.2}
\dr_t I(S;(X_i(t),X_j(t)))_{\mid t = 0} = \frac 1 2 \E \Ll[ \Ll| \bar f_i(S) \Rr| ^2 \Rr]  + \frac 1 2 \E \Ll[ \Ll| \bar f_j(S) \Rr| ^2 \Rr] ,
\end{equation}
and
\begin{multline}  
\label{e.second.order.2}
\dr_t^2 I(S;(X_i(t),X_j(t)))_{\mid t = 0}  
\\
= \frac 1 2 \Ll| \E \Ll[ \bar f_i(S) \bar f_i(S)^* \Rr] \Rr|^2 + \frac 1 2 \Ll| \E \Ll[ \bar f_j(S) \bar f_j(S)^* \Rr] \Rr|^2  +  \Ll| \E \Ll[ \bar f_i(S) \bar f_j(S)^* \Rr] \Rr|^2 .
\end{multline}
By Lemma~\ref{l.i.identity}, we have that 
\begin{equation*}  
I(X_i(t);X'_j(t)) = I(S;(X_i(t);X'_j(t))) - I(S;X_i(t)) - I(S;X_j(t)).
\end{equation*}
A Taylor expansion near $t = 0$ of this identity, combined with the expressions of the derivatives obtained above, therefore yields \eqref{e.psd.lim}. To see that the matrix in \eqref{e.psd} is positive semidefinite, let us denote by $\td S$ an independent copy of the random variable $S$. Writing $\cdot$ for the entrywise scalar product between vectors or matrices, we have for every $\al_1, \ldots, \al_n \in \R$ that
\begin{align*}  
\sum_{i,j = 1}^n \al_i \al_j \Ll| \E \Ll[ \bar f_i(S) \bar f_j(S)^* \Rr]  \Rr|^2  
& = \sum_{i,j = 1}^n \al_i \al_j \E \Ll[ \big(\bar f_i(S) \bar f_j(S)^*\big) \cdot \big(\bar f_i(\td S) \bar f_j(\td S)^* \big)\Rr]  
\\
& = \sum_{i,j = 1}^n \al_i \al_j \E \Ll[ \big(\bar f_i(S) \cdot \bar f_i(\td S)\big) \big(\bar f_j(S) \cdot \bar f_j(\td S)\big)\Rr]  
\\
& = \E \Ll[  \Ll( \sum_{i = 1}^n \al_i \bar f_i(S) \cdot \bar f_i(\td S) \Rr) ^2 \Rr] \ge 0.
\end{align*}
This completes the proof of Proposition~\ref{p.psd}.
\end{proof}

%
%
%
%

\section{Consequences for community detection}
\label{s.community}

Our initial motivation for exploring questions such as \ref{e.question} comes from the study of the mutual information of a problem of community detection in the stochastic block model. In the notation of \cite{sparse_prob}, we specialize to the choice of parameters $p = \frac 1 2$, $t = 0$, $\mu = t_1 \delta_1 + t_2 \delta_{-1}$, with $t_1, t_2 \ge 0$, so that the mutual information considered there simplifies and matches the assumptions of Proposition~\ref{p.taylor.expansion}, as we explain now. First, we sample $S$ as a Bernoulli random variable with parameter $1/2$ (this is one coordinate of $\sigma^*$ in the notation of \cite{sparse_prob}, except that we reparametrize this random variable taking values $\{-1,1\}$ into $S$ taking values in $\{0,1\}$ for notational consistency). Conditionally on $S$, we let $(X_1^{(\ell)}, X_2^{(\ell)})_{\ell \ge 1}$ be independent random variables, with $X_1^{(\ell)}$ sampled according to $P_1$ and $X_2^{(\ell)}$ sampled according to $P_2$, where the channels $P_1$ and $P_2$ are defined by 
\begin{equation}  
\label{e.channel.defs}
P_1(\cdot \mid s) = \Ber(p_s/N)\quad 
\text{ and }
\quad
P_2(\cdot \mid s) = \Ber(q_s/N)\quad \quad (s \in \{0,1\}),
\end{equation}
and $p_0, p_1, q_0, q_1 \in [0,\infty)$ are such that $p_0 = q_1$ and $p_1 = q_0$ (in the notation of \cite{sparse_prob}, we have $p_1 = q_0 = c + \Delta$, and $p_0 = q_1 = c-\Delta$, with the identification that $\sigma^* = 1$ and $-1$ correspond to $S = 1$ and $0$ respectively). While we will not always say it explicitly, we always understand that $N$ is taken sufficiently large that the quantities $p_s/N$ and $q_s/N$ appearing in \eqref{e.channel.defs} belong to the interval $[0,1]$. Finally, we let $\Pi^{(1)}_{N t_1}$ and $\Pi^{(2)}_{N t_2}$ be two independent Poisson random variables of parameters $N t_1$ and $N t_2$ respectively, independent of the all other random variables. With all these choices, and using the Poisson coloring theorem (see for instance \cite[Chapter~5]{Kingman}) we get that the mutual information studied in \cite{sparse_prob} simplifies into
\begin{equation*}  
\mcl I_N(t_1,t_2) := I\Ll(S;\Ll((X_1^{(\ell)})_{\ell \le \Pi^{(1)}_{N t_1}}, (X_2^{(\ell)})_{\ell \le \Pi^{(2)}_{N t_2}}\Rr)\Rr).
\end{equation*}
Although this is not apparent in the notation, we emphasize that the laws of $X_1^{(\ell)}$ and $X_2^{(\ell)}$ depend on $N$. As shown in \cite[Lemma~3.1]{sparse_prob}, the function $\mcl I_N$ converges pointwise to a limit, which we denote by $\mcl I_\infty$. 
\begin{proposition}[Breakdown of concavity of mutual information]
\label{p.non-concave}
For every $N \in \N \cup \{\infty\}$, the entries of the Hessian of the mapping $(t_1,t_2) \mapsto \mcl I_N(t_1,t_2)$ are nonpositive. However, in the regime of finite $N$ going to infinity, we have
\begin{equation}  
\label{e.non-concave.1}
\Ll(\dr_{t_1}^2 \mcl I_N + \dr_{t_2}^2 \mcl I_N - 2 \dr_{t_1} \dr_{t_2} \mcl I_N\Rr)(0,0) \ge \frac{(p_0 - p_1)^6}{6(p_0 + p_1)^4} + o(1) 
\end{equation}
as well as
\begin{equation}  
\label{e.non-concave.2}
\Ll(\dr_{t_1}^2 \mcl I_\infty + \dr_{t_2}^2 \mcl I_\infty - 2 \dr_{t_1} \dr_{t_2} \mcl I_\infty\Rr)(0,0) \ge \frac{(p_0 - p_1)^6}{6(p_0 + p_1)^4}.
\end{equation}
In particular, for every sufficiently large $N \in \N \cup \{\infty\}$, the mapping $(t_1,t_2) \mapsto \mcl I_N(t_1,t_2)$ is not concave.
\end{proposition}
\begin{proof}
We decompose the proof into four steps. 

\smallskip

\noindent \emph{Step 1.} In this step, we derive convenient representations for the second derivatives of~$\mcl I_N$, for finite $N$. For every $t \ge 0$ and $L \in \Z_+$, we denote
\begin{equation*}  
\pi(t,L) := e^{-t} \frac{t^L}{L!}.
\end{equation*}
With the understanding that $\pi(t,-1) = 0$, we have the identity
\begin{equation}  
\label{e.poisson.identity}
\dr_t \pi(t,L) = \pi(t,L-1) - \pi(t,L).
\end{equation}
In order to lighten the calculations, we also introduce the shorthand notation
\begin{equation*}  
I_N(L_1,L_2) = I\Ll(S;\Ll((X_1^{(\ell)})_{\ell \le L_1}, (X_2^{(\ell)})_{\ell \le L_2}\Rr) \Rr).
\end{equation*}
We start by observing that
\begin{equation*}  
\mcl I_N(t_1, t_2) = \sum_{L_1, L_2  = 0}^{+\infty} \pi(Nt_1, L_1) \pi(Nt_2, L_2) I_N(L_1,L_2).
\end{equation*}
The identity \eqref{e.poisson.identity} yields that
\begin{equation*}  
\dr_{t_1} \mcl I_N(t_1,t_2) = N \sum_{L_1,L_2 = 0}^{+\infty} \pi(Nt_1, L_1) \pi(Nt_2, L_2) \Ll(I_N(L_1+1, L_2) - I_N(L_1,L_2) \Rr),
\end{equation*}
and thus
\begin{multline}  
\label{e.expr.second.der}
\dr_{t_1}^2 \mcl I_N(t_1,t_2) = N^2 \sum_{L_1,L_2 = 0}^{+\infty} \pi(Nt_1, L_1) \pi(Nt_2, L_2) \\ \Ll(I_N(L_1+2, L_2) - 2I_N(L_1+1,L_2) + I_N(L_1,L_2)\Rr).
\end{multline}
A similar expression can be obtained for $\dr_{t_2}^2 \mcl I_N$, with the finite-difference operation acting on the variable $L_2$ in place of $L_1$. The cross-derivative takes the form
\begin{multline}  
\label{e.expr.cross.der}
\dr_{t_1} \dr_{t_2} \mcl I_N(t_1,t_2) = N^2 \sum_{L_1,L_2 = 0}^{+\infty} \pi(Nt_1, L_1) \pi(Nt_2, L_2) \\ \Ll(I_N(L_1+1, L_2+1) - I_N(L_1+1,L_2) - I_N(L_1,L_2 +1) + I_N(L_1,L_2)\Rr).
\end{multline}

\smallskip

\noindent \emph{Step 2.} In this step, we show that the entries of the Hessian of $\mcl I_N$ are nonpositive. Since this property can be understood in a weak sense, or in terms of the signs of certain finite differences, it suffices to show its validity for finite $N$. From the expressions of the second derivatives obtained in the previous step, we see that it suffices to show that, for every $N \in \N, L_1, L_2 \in \Z_+$,
\begin{equation}
\label{e.nonpos.1}
I_N(L_1+2, L_2) - 2I_N(L_1+1,L_2) + I_N(L_1,L_2) \le 0, 
\end{equation}
\begin{equation}
\label{e.nonpos.2}
I_N(L_1, L_2+2) - 2I_N(L_1,L_2+1) + I_N(L_1,L_2) \le 0,
\end{equation}
and 
\begin{equation}
\label{e.nonpos.3}
I_N(L_1+1, L_2+1) - I_N(L_1+1,L_2) - I_N(L_1,L_2 +1) + I_N(L_1,L_2) \le 0.
\end{equation}
We only show the validity of \eqref{e.nonpos.3}, the arguments for \eqref{e.nonpos.1} and \eqref{e.nonpos.2} being similar. In order to lighten the notation, we write 
\begin{equation*}  
Z := \Ll( (X_1^{(\ell)})_{\ell \le L_1}, (X_2^{(\ell)})_{\ell \le L_2} \Rr) .
\end{equation*}
By the chain rule for mutual information, we have
\begin{align*}  
I_N(L_1+1, L_2+1) 
& = I\Ll(S;\Ll((X_1^{(\ell)})_{\ell \le L_1+1}, (X_2^{(\ell)})_{\ell \le L_2+1}\Rr) \Rr)
\\
& = I\Ll(S;\Ll(X_1^{(L_1+1)}, X_2^{(L_2+1)}\Rr) \mid Z \Rr) + I(S;Z),
\end{align*}
and similarly, 
\begin{equation*}  
I_N(L_1+1,L_2) = I\Ll(S;X_1^{(L_1+1)} \mid Z \Rr) + I(S;Z) ,
\end{equation*}
and
\begin{equation*}  
I_N(L_1,L_2+1) = I\Ll(S;X_2^{(L_2+1)} \mid Z \Rr) + I(S;Z).
\end{equation*}
Showing \eqref{e.nonpos.3} is thus equivalent to showing that
\begin{equation}
\label{e.nonpos.3bis}
I\Ll(S;\Ll(X_1^{(L_1+1)}, X_2^{(L_2+1)}\Rr) \mid Z \Rr) - I\Ll(S;X_1^{(L_1+1)} \mid Z \Rr)  -  I\Ll(S;X_2^{(L_2+1)} \mid Z \Rr)  \le 0.
\end{equation}
We use again the chain rule of mutual information to write
\begin{align*}  
I\Ll(S;\Ll(X_1^{(L_1+1)}, X_2^{(L_2+1)}\Rr) \mid Z \Rr) 
& = I\Ll(S;X_1^{(L_1+1)} \mid Z \Rr) + I\Ll(S;X_2^{(L_2+1)} \mid X_1^{(L_1+1)}, Z \Rr).
\end{align*}
The last term of the identity above can be rewritten as
\begin{multline}  
\label{e.chainagain}
 I\Ll(X_2^{(L_2+1)} ; \Ll(S,X_1^{(L_1+1)}\Rr) \mid  Z \Rr) - I\Ll(X_2^{(L_2+1)};X_1^{(L_1+1)} \mid  Z \Rr) 
\\
 \quad = I\Ll(S;X_2^{(L_2+1)}  \mid  Z \Rr) + I\Ll(X_1^{(L_1+1)}; X_2^{(L_2+1)} \mid S, Z\Rr)- I\Ll(X_1^{(L_1+1)};X_2^{(L_2+1)} \mid  Z \Rr) .
\end{multline}
Conditionally on $S$, the random variables $(X_1^{(L_1+1)}, X_2^{(L_2+1)}, Z)$ are independent, and thus the second term on the right side of \eqref{e.chainagain} is zero. Combining these identities, we obtain that the left side of \eqref{e.nonpos.3bis} equals $-I\Ll(X_1^{(L_1+1)};X_2^{(L_2+1)} \mid  Z \Rr) $, which is indeed nonpositive.

\smallskip

\noindent \emph{Step 3.} In this step, we show the validity of \eqref{e.non-concave.1}, and thus deduce the non-concavity of $\mcl I_N$ for every $N$ sufficiently large and finite. Using the expressions for the second derivative obtained in \eqref{e.expr.second.der} and \eqref{e.expr.cross.der}, we can write
\begin{align*}  
& \begin{pmatrix}  
1 \\ -1 
\end{pmatrix}
\cdot
\begin{pmatrix}  
\dr^2_{t_1} \mcl I_N(0,0) & \dr_{t_1} \dr_{t_2} \mcl I_N(0,0) \\ \dr_{t_1} \dr_{t_2} \mcl I_N(0,0) & \dr_{t_2}^2 \mcl I_N(0,0)
\end{pmatrix}
\begin{pmatrix}  
1 \\ -1 
\end{pmatrix} 
& 
\\
& \qquad = N^2\Ll[I(S;(X_1^{(1)},X_1^{(2)})) + I(S;(X_2^{(1)},X_2^{(2)})) - 2 I(S;(X_1^{(1)},X_2^{(1)}))\Rr]
\\
 & \qquad = N^2 \Ll[2 I(X_1^{(1)};X_2^{(1)}) - I(X_1^{(1)};X_1^{(2)}) - I(X_2^{(1)};X_2^{(2)}) \Rr],
\end{align*}
where we used Lemma~\ref{l.i.identity} in the last step. Proposition~\ref{p.taylor.expansion} ensures that, for finite $N$ going to infinity, we have
\begin{equation*}  
N^2 \Ll[2 I(X_1^{(1)};X_2^{(1)}) - I(X_1^{(1)};X_1^{(2)}) - I(X_2^{(1)};X_2^{(2)}) \Rr] \ge \frac{(p_0 - p_1)^6}{6(p_0 + p_1)^4} + o(1),
\end{equation*}
which gives the desired result.

\smallskip 

\noindent \emph{Step 4.} In this last step, we show the validity of \eqref{e.non-concave.2}. 
Instead of trying to justify that the second derivatives of $\mcl I_N$ converge to those of $\mcl I_\infty$, we simply borrow from \cite{sparse_prob} an explicit expression for $\mcl I_\infty$, and observe that it satisfies \eqref{e.non-concave.2} by calculating its derivatives.  We recall that $p_1 = q_0$ corresponds to $c+\Delta$ in the notation of \cite{sparse_prob}, while $p_0 = q_1$ corresponds to $c-\Delta$ in the notation of \cite{sparse_prob}. The statement of \cite[Lemma~3.1]{sparse_prob} involves two Poisson point processes, denoted by $\Pi_+$ and $\Pi_-$ there, and which in our present context can be represented as $\Pi^{(1)}_{p_1 t_1} \delta_1 + \Pi^{(2)}_{p_0 t_2} \delta_{-1}$ and $\Pi^{(1)}_{p_0 t_1} \delta_1 + \Pi^{(2)}_{p_1 t_2} \delta_{-1}$ respectively. The quantity $\mu[-1,1]\E x_1$ appearing in \cite[Lemma~3.1]{sparse_prob} translates into $t_1-t_2$ in our context. The function that is denoted by $\psi(\mu)$ in the notation of \cite[Lemma~3.1]{sparse_prob} becomes, in our current setting, the function
\begin{align*}  
\psi(t_1,t_2) & :=   -(t_1+t_2)\frac{p_1 + p_0}{2} 
\\
& \qquad 
+ \frac 1 2 \E \log\Ll[ \frac 1 2 e^{-\frac{(t_1-t_2)(p_1-p_0)}{2}} p_1^{\Pi^{(1)}_{p_1 t_1}}p_0^{\Pi^{(2)}_{p_0 t_2}} + \frac 1 2 e^{\frac{(t_1-t_2)(p_1-p_0)}{2}}  p_0^{\Pi^{(1)}_{p_1 t_1}}p_1^{\Pi^{(2)}_{p_0 t_2}}
 \Rr]
\\
& \qquad 
 + \frac 1 2 \E \log \Ll[\frac 1 2 e^{-\frac{(t_1-t_2)(p_1-p_0)}{2}} p_1^{\Pi^{(1)}_{p_0 t_1}}p_0^{\Pi^{(2)}_{p_1t_2}} + \frac 1 2 e^{\frac{(t_1-t_2)(p_1-p_0)}{2}} p_0^{\Pi^{(1)}_{p_0 t_1}}p_1^{\Pi^{(2)}_{p_1t_2}} \Rr]  .
\end{align*}
Arguing as for \cite[(1.16)-(1.17)]{sparse_prob}, one can check that the mutual information $\mcl I_N(t_1,t_2)$ is obtained as a simple (and convergent as $N \to \infty$) linear term in $(t_1,t_2)$, minus a function, denoted by $\psi_N(\mu)$ in the notation of \cite[Lemma~3.1]{sparse_prob}, that converges to $\psi(t_1,t_2)$. In order to show that the mapping $(t_1,t_2) \mapsto \mcl I_\infty(t_1,t_2)$ is not concave, it thus suffices to show that the mapping
\begin{align*}  
(t_1,t_2) \mapsto \phi(t_1,t_2) 
  := & \frac 1 2 \E \log\Ll[ \frac 1 2 e^{-\frac{(t_1-t_2)(p_1-p_0)}{2}} p_1^{\Pi^{(1)}_{p_1 t_1}}p_0^{\Pi^{(2)}_{p_0 t_2}} + \frac 1 2 e^{\frac{(t_1-t_2)(p_1-p_0)}{2}}  p_0^{\Pi^{(1)}_{p_1 t_1}}p_1^{\Pi^{(2)}_{p_0 t_2}}
 \Rr]
\\
 + & \frac 1 2 \E \log \Ll[\frac 1 2 e^{-\frac{(t_1-t_2)(p_1-p_0)}{2}} p_1^{\Pi^{(1)}_{p_0 t_1}}p_0^{\Pi^{(2)}_{p_1t_2}} + \frac 1 2 e^{\frac{(t_1-t_2)(p_1-p_0)}{2}} p_0^{\Pi^{(1)}_{p_0 t_1}}p_1^{\Pi^{(2)}_{p_1t_2}} \Rr]  
\end{align*}
is not convex. For every $s \in \R$ and integers $L_1, L_2 \ge 0$, we denote
\begin{equation*}  
J(s,L_1, L_2) := \log\Ll[ \frac 1 2 e^{-\frac{s(p_1-p_0)}{2}} p_1^{L_1}p_0^{L_2} + \frac 1 2 e^{\frac{s(p_1-p_0)}{2}}  p_0^{L_1}p_1^{L_2} \Rr],
\end{equation*}
and observe that 
\begin{align*}  
\phi(t_1,t_2) = \frac 1 2 \sum_{L_1, L_2 = 0}^{+\infty} \Ll(\pi(p_1 t_1, L_1) \pi(p_0 t_2, L_2)  + \pi(p_0 t_1, L_1) \pi(p_1 t_2, L_2)  \Rr) J(t_1-t_2,L_1, L_2). 
\end{align*}
Using the  identity \eqref{e.poisson.identity}, we write 
\begin{align*}
    \dr_{t_1} \phi (t_1, t_2) & = \frac 1 2 \sum_{L_1, L_2 = 0}^{+\infty}J(t_1 - t_2, L_1, L_2) \Big [  \Ll ( \pi(p_1 t_1, L_1 - 1) - \pi(p_1 t_1, L_1)  \Rr) p_1 \pi(p_0 t_2, L_2)  \\
    &\qquad\qquad +  \Ll ( \pi(p_0 t_1, L_1 - 1) - \pi(p_0 t_1, L_1)  \Rr) p_0 \pi(p_1 t_2, L_2) \Big]   \\
    &\qquad + \dr_{t_1} J(t_1 - t_2, L_1, L_2) \Big [ \pi(p_1 t_1, L_1) \pi(p_0 t_2, L_2) + \pi(p_0 t_1, L_1) \pi(p_1 t_2, L_2) \Big ] ,
\end{align*}
and 
\begin{align*}
    \dr^2_{t_1} \phi (t_1, t_2) & = \frac 1 2 \sum_{L_1, L_2 = 0}^{+\infty} J(t_1 - t_2, L_1, L_2) \\
    & \qquad \Big  [ \Ll ( \pi(p_1 t_1, L_1 - 2) - 2\pi(p_1 t_1, L_1 - 1) + \pi(p_1 t_1, L_1) \Rr) p_1^2 \pi(p_0 t_2, L_2)  \\
    &\qquad\qquad +  \Ll ( \pi(p_0 t_1, L_1 - 2) - 2\pi(p_0 t_1, L_1 - 1) + \pi(p_0 t_1, L_1) \Rr) p_0^2 \pi(p_1 t_2, L_2) \Big]   \\
    &\qquad +\dr_{t_1} J(t_1 - t_2, L_1, L_2) \Big [  \Ll ( \pi(p_1 t_1, L_1 - 1) - \pi(p_1 t_1, L_1)  \Rr) p_1 \pi(p_0 t_2, L_2)  \\
    &\qquad\qquad +  \Ll ( \pi(p_0 t_1, L_1 - 1) - \pi(p_0 t_1, L_1)  \Rr) p_0 \pi(p_1 t_2, L_2) \Big] \\
    &\qquad + \dr^2_{t_1} J(t_1 - t_2, L_1, L_2) \Big [ \pi(p_1 t_1, L_1) \pi(p_0 t_2, L_2) + \pi(p_0 t_1, L_1) \pi(p_1 t_2, L_2) \Big ] .
\end{align*}
Similar calculations yield 
\begin{align*}
    \dr_{t_1} \dr_{t_2} \phi(t_1, t_2) & = \frac 1 2 \sum_{L_1, L_2 = 0}^{+\infty}  J(t_1 - t_2, L_1, L_2) \\
    & \qquad p_0 p_1  \Big  [ \Ll ( \pi(p_1 t_1, L_1 - 1) - \pi(p_1 t_1, L_1) \Rr) \Ll(  \pi(p_0 t_2, L_2 - 1) - \pi(p_0 t_2, L_2) \Rr ) \\
    &\qquad\qquad + \Ll ( \pi(p_0 t_1, L_1 - 1) - \pi(p_0 t_1, L_1) \Rr) \Ll(  \pi(p_1 t_2, L_2 - 1) - \pi(p_1 t_2, L_2) \Rr ) \Big ] \\
    &\qquad + \dr_{t_1} J(t_1 - t_2, L_1, L_2) \Big [p_0 \pi(p_1 t_1, L_1) \Ll(  \pi(p_0 t_2, L_2 - 1) - \pi(p_0 t_2, L_2) \Rr )  \\
    &\qquad\qquad + p_1 \pi(p_0 t_1, L_1) \Ll(  \pi(p_1 t_2, L_2 - 1) - \pi(p_1 t_2, L_2) \Rr ) \Big] \\
    &\qquad + \dr_{t_2} J(t_1 - t_2, L_1, L_2) \Big [ p_1 \Ll ( \pi(p_1 t_1, L_1 - 1) - \pi(p_1 t_1, L_1) \Rr) \pi(p_0 t_2, L_2) \\
    &\qquad\qquad + p_0 \Ll ( \pi(p_0 t_1, L_1 - 1) - \pi(p_0 t_1, L_1) \Rr) \pi(p_1 t_2, L_2) \Big] \\
    &\qquad + \dr_{t_1} \dr_{t_2} J(t_1 - t_2, L_1, L_2) \Big [ \pi(p_1 t_1, L_1) \pi(p_0 t_2, L_2) + \pi(p_0 t_1, L_1) \pi(p_1 t_2, L_2)  \Big].
\end{align*}
We get a similar expression for $\dr_{t_2}^2 \phi (t_1, t_2) $ as for $\dr_{t_1}^2 \phi (t_1, t_2) $:
\begin{align*}
    \dr^2_{t_2} \phi (t_1, t_2) & = \frac 1 2 \sum_{L_1, L_2 = 0}^{+\infty} J(t_1 - t_2, L_1, L_2) \\
    & \qquad \Big  [ p_0^2 \pi(p_0 t_1, L_1) \Ll ( \pi(p_0 t_2, L_2 - 2) - 2\pi(p_0 t_2, L_2 - 1) + \pi(p_0 t_2, L_2) \Rr)   \\
    &\qquad\qquad +   p_1^2 \pi(p_0 t_1, L_1) \Ll (\pi(p_1 t_2, L_2 - 2) - 2\pi(p_1 t_2, L_2 - 1) + \pi(p_1 t_2, L_2) \Rr)  \Big]   \\
    &\qquad +\dr_{t_2} J(t_1 - t_2, L_1, L_2)  \Big [ p_0 \pi(p_1 t_1, L_1)  \Ll ( \pi(p_0 t_2, L_2 - 1) - \pi(p_0 t_2, L_2)  \Rr)  \\
    &\qquad\qquad +  p_1 \pi(p_0 t_1, L_1) \Ll ( \pi(p_1 t_2, L_2 - 1) - \pi(p_1 t_2, L_2)  \Rr)  \Big] \\
    &\qquad + \dr^2_{t_2} J(t_1 - t_2, L_1, L_2) \Big [ \pi(p_1 t_1, L_1) \pi(p_0 t_2, L_2) + \pi(p_0 t_1, L_1) \pi(p_1 t_2, L_2) \Big ] .
\end{align*}

We are interested in the value of $2 \dr_{t_1} \dr_{t_2} \phi(t_1,t_2) - \dr_{t_1}^2 \phi(t_1,t_2) - \dr_{t_2}^2 \phi(t_1,t_2)$ at $(t_1, t_2) = (0,0)$. Hence, we use the Taylor expansion of $J(t_1 - t_2, L_1, L_2) $ for arbitrary $L_1, L_2 \in \Z_+$ at $t_1 = t_2 = 0$ to get expressions for derivatives up to the second order
\begin{multline}  \label{e.taylor_J}
J(t_1 - t_2, N_1, N_2) = \log \Ll( \frac{p_1^{N_1} p_0^{N_2} + p_0^{N_1} p_1^{N_2}}{2} \Rr) - (t_1 - t_2)\frac{p_1 - p_0}{2} \frac{p_1^{N_1} p_0^{N_2} - p_0^{N_1} p_1^{N_2}}{p_1^{N_1} p_0^{N_2} + p_0^{N_1} p_1^{N_2}} 
\\
+ \frac 1 2 (t_1 - t_2)^2 {(p_1 - p_0)^2} \frac{p_1^{N_1+N_2} p_0^{N_1 + N_2}}{(p_1^{N_1} p_0^{N_2} + p_0^{N_1} p_1^{N_2})^2} + O((t_1-t_2)^3).
\end{multline}
Further, note that $\pi(0, L) = 0$ for all $L \ge 1$ and $\pi(0,0)=1$. With this observation, we can write the second derivatives of $\phi$ as sum of a few terms.
\begin{equation}\label{e.phi_derivatives_diff}
    \begin{split}
    &2 \dr_{t_1} \dr_{t_2} \phi(0,0) - \dr_{t_1}^2 \phi(0,0) - \dr_{t_2}^2 \phi(0,0) \\
    &\qquad = \frac 1 2 \Big ( - (p_0-p_1)^2 J(0,0,0) + 2( p_0-p_1)^2 J(0, 1, 0) + 2(p_0-p_1)^2 J(0,0,1)\\
    &\qquad\qquad + 4 p_0 p_1  J(0, 1, 1) - (p_0^2 + p_1^2) J(0, 2, 0) - (p_0^2 + p_1^2) J(0, 0, 2)\Big) \\
    &\qquad\qquad - (p_0 + p_1) \dr_{t_1}J(0, 1, 0) + (p_0 + p_1) \dr_{t_1}J(0, 0, 1) \\
    &\qquad\qquad + (p_0 + p_1) \dr_{t_2}J(0, 1, 0) - (p_0 + p_1) \dr_{t_2}J(0, 0, 1) \\
    &\qquad\qquad +  2 \dr_{t_1} \dr_{t_2} J(0,0,0) - \dr_{t_1}^2 J(0,0,0) - \dr_{t_2}^2 J(0,0,0) .
    \end{split}
\end{equation}
From \eqref{e.taylor_J} we get 
\begin{equation*}
    2 \dr_{t_1} \dr_{t_2} J(0,0,0) - \dr_{t_1}^2 J(0,0,0) - \dr_{t_2}^2 J(0,0,0) = - (p_0 - p_1)^2.
\end{equation*}
Using that $\dr_{t_1} J (0,1,0) = - \dr_{t_1} J (0,0,1) = -\nicefrac{(p_0 - p_1)^2}{2(p_0 + p_1)} $, we obtain that
\begin{equation*}
    - (p_0 + p_1) \dr_{t_1}J(0, 1, 0) + (p_0 + p_1) \dr_{t_1}J(0, 0, 1) = (p_0 - p_1)^2.
\end{equation*}
Similarly, $(p_0 + p_1) \dr_{t_2}J(0, 1, 0) - (p_0 + p_1) \dr_{t_2}J(0, 0, 1) = (p_0  - p_1)^2$. It remains to compute the first terms of \eqref{e.phi_derivatives_diff} that do not contain derivatives of $J$. 
\begin{equation}\label{e.phi_derivatives_diff_first_term}
\begin{split}
    &- (p_0-p_1)^2 J(0,0,0) + 2( p_0-p_1)^2 J(0, 1, 0) + 2(p_0-p_1)^2 J(0,0,1)\\
    &\qquad\qquad + 4 p_0 p_1  J(0, 1, 1) - (p_0^2 + p_1^2) J(0, 2, 0) - (p_0^2 + p_1^2)J(0,0,2)\\ 
    &\qquad= 4 (p_0 - p_1)^2 \log\Ll( \frac{p_0 + p_1}{2}  \Rr) + 4 p_0 p_1 \log (p_0 p_1) - 2 (p_0+p_1)^2 \log \Ll ( \frac{p_0^2 + p_1^2 }{2} \Rr).
\end{split}
\end{equation}
Combining all together and rearranging terms in \eqref{e.phi_derivatives_diff_first_term}, we get 
\begin{multline*}
    2 \dr_{t_1} \dr_{t_2} \phi(0,0) - \dr_{t_1}^2 \phi(0,0) - \dr_{t_2}^2 \phi(0,0) = \\(p_0 - p_1)^2 + 2 p_0 p_1 \log \Ll (1 - \frac{(p_0 - p_1)^2}{ (p_0 + p_1)^2} \Rr ) -  ( p_0^2 + p_1^2) \log \Ll (1 + \frac{(p_0 - p_1)^2}{ (p_0 + p_1)^2} \Rr ).
\end{multline*}
The right-hand side of the above equality coincides with the right-hand side of \eqref{e.mutual_info_lbound}, with the $o(1)$ term taken out. As was shown in the proof of \Cref{p.taylor.expansion}, this expression is lower-bounded by $\nicefrac{(p_0 - p_1)^6}{6 (p_0 + p_1)^4}$.  
\end{proof}
\bibliographystyle{abbrv}
\bibliography{mutual_information}

\end{document}